\documentclass[letterpaper, 10 pt, conference]{ieeeconf}  % Comment this line out if you need a4paper

\IEEEoverridecommandlockouts                              % This command is only needed if 
\usepackage{amsmath} % assumes amsmath package installed
\DeclareMathOperator*{\argmax}{arg\,max}

\usepackage{amsthm}
\usepackage{graphicx}
\usepackage{xcolor}
\usepackage{cite}

\usepackage{amssymb}  % assumes amsmath package installed

\title{\LARGE \bf
On Connections between Opacity and Security in Linear Systems
}

\author{Varkey M. John and Vaibhav Katewa% <-this % stops a space
\thanks{V. M. John is with the Department of ECE at the Indian Institute of Science (IISc) Bangalore. He is supported
by a fellowship grant from the Cisco Centre for Networked Intelligence  at IISc. V. Katewa is with the Robert Bosch Center for CPS and Department
of ECE at IISc Bangalore. Email IDs:
        {\tt\small \{varkeym@iisc.ac.in,
vkatewa@iisc.ac.in\}}}}% <-this % stops a space

\newtheorem{theorem}{Theorem}
\newtheorem{corollary}{Corollary}
\newtheorem{lemma}{Lemma}

\theoremstyle{definition}
\newtheorem{definition}{Definition}
\newtheorem{example}{Example}
\theoremstyle{remark}
\newtheorem{remark}{Remark}
\theoremstyle{definition}

\newcommand\numberthis{\addtocounter{equation}{1}\tag{\theequation}}

\begin{document}
\maketitle
\thispagestyle{empty}
\pagestyle{empty}

%%%%%%%%%%%%%%%%%%%%%%%%%%%%%%%%%%%%%%%%%%%%%%%%%%%%%%%%%%%%%%%%%%%%%%%%%%%%%%%%
\begin{abstract}
Opacity and attack detectability are important properties for any system as they allow the states to remain private and malicious attacks to be detected, respectively. In this paper, we show that a fundamental trade-off exists between these properties for a linear dynamical system, in the sense that if an opaque system is subjected to attacks, all attacks cannot be detected. We first characterize the opacity conditions for the system in terms of its weakly unobservable subspace (WUS) and show that the number of opaque states is proportional to the size of the WUS.  Further, we establish conditions under which increasing the opaque sets also increases the set of undetectable attacks. This highlights a fundamental trade-off between security and privacy. We demonstrate application of our results on a remotely controlled automotive system.
\end{abstract}
%%%%%%%%%%%%%%%%%%%%%%%%%%%%%%%%%%%%%%%%%%%%%%%%%%%%%%%%%%%%%%%%%%%%%%%%%%%%%%%%
\section{Introduction}
Cyber-physical attacks have become significantly prevalent in recent years, including the Stuxnet attack (2010) and the Maroochy Shire attack (2000) \cite{Stuxnet_IEEESpectrum_2013}, \cite{MaroochyShire_MIT_2017}. Due to such vulnerabilities, there has been a larger thrust in the last decade to enable these systems to detect attacks upfront \cite{AttackDetection_IEEETAC_2013,DynamicDetection_IEEETAC_2017}. Such detection mechanisms are especially relevant since traditional cyber-security solutions cannot be used to detect real-time physics-based attacks.\par
In parallel, increased demand for privacy has led to a focus on keeping information from Cyber-Physical Systems (CPS) confidential. In particular, the notion of opacity, which was first considered in the computer science literature (with discrete events) \cite{UnificationOpacity_ITS_2004,OpacityPetriNets_ENTCS_2005}, has been applied to CPS and dynamical systems (with continuous state space) in recent years \cite{OpacityLinearSystems_IEEETAC_2020,ApproxOpacity_IEEETAC_2021,OpacityRobControl_IEEETAC_2019,OpacityDSE_Automatica_2022}. Informally, opacity requires that same outputs should be produced by a secret as well as a non-secret initial states. This prevents an eavesdropper to distinguish whether the system was initialized in a secret or non-secret state based on the outputs. In \cite{OpacityLinearSystems_IEEETAC_2020}, the authors developed the notion of opacity for linear dynamical systems and showed its relation to other system properties like output controllability. A relaxed notion of ``approximate opacity" was developed in \cite{ApproxOpacity_IEEETAC_2021}, where the outputs from secret and non-secret states were allowed to be ``close" to each other. Algorithms to enforce opacity for robust control and distributed state estimation in linear CPS are proposed in \cite{OpacityRobControl_IEEETAC_2019} and \cite{OpacityDSE_Automatica_2022}, respectively.\par
While research on security and privacy have produced a large spectrum of results individually, studies that assess the impact of security on privacy, and vice-versa, are fairly limited. In \cite{IntegrityDP_ACC_2017}, the authors discuss how differential privacy mechanism can weaken system's security against integrity attacks. The trade-off between local mechanisms of security and privacy in interconnected dynamical systems is analyzed in \cite{SecurityPrivacyInterconn_Automatica_2021}. A game-theoretic approach to the security-privacy trade-off using Quantitative Information Flow theory is investigated in \cite{QIF_IEEETIFAS_2019}. This paper aims to explore and elucidate the fundamental connections between opacity and attack detectability. The main contributions of the paper are: 

\noindent 1. We characterize the relation between opacity and the Weakly Unobservable Subspace (WUS), and use this to derive conditions for opacity of initial states. 

\noindent 2. We show that there exists a trade-off between opacity and attack detectability, and one cannot have an opaque system without making it vulnerable to undetectable attacks. Further, we show that increasing the opaque set also increases the set of undetectable attacks under certain conditions.

%In the first part of our paper, we build up on current works on opacity and we characterize opacity for linear systems from foundational principles. From our study, we discovered that a fundamental connection exists between opacity of a system and its weakly unobservable subspace. Using this connection, opacity was formulated for initial states and sets of initial states. Also, the change in opacity was assessed when different parameters of the system was changed.\par
%In the second part of our work, we seek to supplement the nascent research area of trade-off between security and privacy. In particular, we investigate this in terms of the system's opacity and attack detectability. As far as we know, such an analysis has not been looked into previously. The results of our analysis prove that an opaque (or private) system is inevitably vulnerable to undetectable attacks (thus less secure).\par
The results are discussed in a running example and we illustrate their practical application on a remotely controlled automotive system.
%In the following, we describe the main contributions of our work.\par
%\noindent\emph{Main Contributions:}\par
%\noindent 1. Characterization of the opacity of a system in terms of its weakly unobservable subspace.\par
%\noindent 2. Analysis of the trade-off between opacity and attack detectability in terms of:\par
%a. Existence of opaque sets and undetectable attacks.\par
%b. Change of opaque sets and undetectable attacks with change in initial state set.\\~\par

\noindent {{\bf Notation:}} $\mathcal{R}(A)$ and $\mathcal{N}(A)$ denote the range space and the null space of matrix $A$, respectively. The orthogonal complement of a vector space $V$ is defined as $V^{\perp}=\{w \in \mathbb{R}^n:w\cdot v=0\:\forall \: v \in V\}$, where $w.v$ represents the dot product between $w$ and $v$. For matrix $A$ and set $\mathcal{S}$, $A\mathcal{S}=\{As:s\in\mathcal{S}\}$. Similarly, for matrix $A$ and vector space $V$, $AV=\{Av:v\in V\}$. $\mathcal{S}_1\bigoplus\mathcal{S}_2$ represents the Minkowski sum of sets $\mathcal{S}_1$ and $\mathcal{S}_2$, defined as $\mathcal{S}_1\bigoplus\mathcal{S}_2=\{s_1+s_2|s_1\in \mathcal{S}_1,s_2\in \mathcal{S}_2\}$. $\mathcal{S}_1\backslash\mathcal{S}_2$ denotes the set minus operation, where elements belonging to $\mathcal{S}_1$ and not belonging to $\mathcal{S}_2$ are chosen. $|\mathcal{S}|$ represents the cardinality of the set $\mathcal{S}$. $\mathcal{S}_1\times\mathcal{S}_2$ denotes the cartesian product of $\mathcal{S}_1$ and $\mathcal{S}_2$. $\phi$ denotes the empty set. 

\section{System, Opacity and Attack Models}
\noindent \emph{\bf System Model:} We consider a linear time-invariant system under normal (unattacked) operation (represented by $\Gamma$): 
\begin{align*}
\Gamma\mathpunct{:}\quad\begin{aligned} x(k+1)&=Ax(k)+Bu(k),\\
y(k)&=Cx(k)+Du(k),\end{aligned}\numberthis\label{normal_model}
\end{align*}
where $x\in \mathbb{R}^n, y\in\mathbb{R}^m, u\in \mathbb{R}^p, k\in \mathbb{Z}$ represent the state, output, normal input and time instant, respectively. 
%In this paper, we analyze opacity in terms of the ability of the system to keep the knowledge of the initial states of $\Gamma$ confidential from an eavesdropper who has access to the output values. 
Let $\mathcal{X}_0$ be the set of initial states in which the system is allowed to begin. 
%A secret set $\mathcal{X}_s$ can be made partially or fully opaque by having a non-secret set $\mathcal{X}_{ns}$ for which the system produces output trajectories that are same as some or all of those produced by the system with initial state set $\mathcal{X}_s$. 
Let $U(k)=\begin{bmatrix}u(0)^T & u(1)^T & \ldots & u(k)^T\end{bmatrix}^T$ denote the input sequence (represented as a vector) until time instant $k$. Further, let $Y_{x(0),U(k)}$ denote the output sequence (vector) produced by applying the input sequence $U(k)$ to the initial state $x(0)$. The output sequence can be written as:
\begin{align*}
Y_{x(0),U(k)}=O_k x(0)+F_k^{\Gamma}U(k),\numberthis \label{nop-seq}
\end{align*}
where $O_k$ and $F_k^{\Gamma}$ are extended observability and forced response matrices, respectively, and are given by:
%For $i\in \{s,ns\}$, let $u_i(k)$ and $U_i(k)=\begin{bmatrix}u_i(0)^T & u_i(1)^T & \ldots & u_i(k)^T\end{bmatrix}^T$ denote the input at $k$ and the input sequence until $k$, respectively, on a system which starts at initial state $x_i(0)$. Also, let $\mathcal{U}_i(k)$ represent the complete set of allowable input sequences $U_i(k)$, i.e., $U_i(k)\in \mathcal{U}_i(k)$.
%Therefore, the output trajectories produced by a system, denoted as $Y_{x_i(0),U_i(k)}$, can be expressed in terms of the observability matrix $O_k$ as: $Y_{x_i(0),U_i(k)}=O_kx_i(0)+F_k^{\Gamma}U_i(k)\enskip\forall k\geq 0$,
\begin{align*}
O_k&=\begin{bmatrix}C^T & (CA)^T & \ldots & (CA^k)^T\end{bmatrix}^T,\numberthis \label{obsv}
\end{align*}
\begin{align*}
F_k^{\Gamma}&=
\begin{bmatrix}
D & 0 & \ldots & 0\\
CB & D & \ldots & 0\\
\vdots & \vdots & \ddots & \vdots\\
CA^{k-1}B & CA^{k-2}B & \ldots & D\\
\end{bmatrix} &&\text{ for }k\geq 1,\numberthis \label{ip-op}
\end{align*}
and $F_k^{\Gamma} = D$ for $k=0$. We assume that system $\Gamma$ is observable. 

\noindent \textbf{Opacity Model:} We consider a potential eavesdropper present in the system whose goal is to gain information about the initial state of the system using the outputs. Let $\mathcal{X}_s \subseteq \mathcal{X}_0 $ denote the set of secret initial states  that a system operator wishes to keep private from the eavesdropper. The remaining set of non-secret initial states is denoted by $\mathcal{X}_{ns}= \mathcal{X}_0\backslash\mathcal{X}_s$. Any element of $\mathcal{X}_{ns}$ is not considered sensitive to disclosure. We use $x_s(0)$ and $x_{ns}(0)$ to denote individual elements in $\mathcal{X}_s$ and $\mathcal{X}_{ns}$, respectively. We assume that the eavesdropper has knowledge of the system matrices $A,B,C,D,$ and the initial state sets $\mathcal{X}_{s}$ and $\mathcal{X}_{ns}$. Next, we provide opacity definitions corresponding to system \eqref{normal_model}.\par
\begin{definition}[\emph{Opacity of Initial States}] \label{def:opacity_state} 
A secret initial state $x_s(0)\in \mathcal{X}_s$ is \emph{opaque} with respect to a non-secret initial state set $\mathcal{X}_{ns}^\prime\subseteq\mathcal{X}_{ns}$ if for all $k\geq 0$ the following property holds: for every $U_{s}(k)$, there exist $x_{ns}(0)\in \mathcal{X}_{ns}^{\prime}$ and $U_{ns}(k)$ such that $Y_{x_s(0),U_s(k)}=Y_{x_{ns}(0),U_{ns}(k)}$. We denote this relation by $x_s(0)\xrightarrow{\text{o}}\mathcal{X}_{ns}^\prime$ and sometimes use the term opaque for such $x_s(0)$. We call $x_s(0)$ \emph{transparent} if it is not opaque. \hfill $\square$
\end{definition}

The opacity definition implies that the same output sequence can result from either a secret or a non-secret initial state (with appropriate control input sequences). Therefore, the eavesdropper who observes the output sequence cannot distinguish whether the system started from a secret or non-secret initial state (assuming it does not have access to input sequences). This makes the secret initial state opaque. Next, we present opacity definitions for sets.

\begin{definition}[\emph{Strong Opacity of Sets}] 
The secret initial state set $\mathcal{X}_s$ is \emph{strongly opaque} with respect to non-secret initial state set $\mathcal{X}_{ns}^\prime\subseteq\mathcal{X}_{ns}$, if for \emph{every} $x_s(0)\in\mathcal{X}_s$, it holds that $x_s(0)\xrightarrow{\text{o}}\mathcal{X}_{ns}^\prime$. We denote this relation by $\mathcal{X}_{s}\xrightarrow{\text{s.o.}}\mathcal{X}_{ns}^\prime$ and sometimes use the term (strongly) opaque for such $\mathcal{X}_s$.\hfill $\square$
\end{definition}

\begin{definition}[\emph{Weak Opacity of Sets}] The secret initial state set $\mathcal{X}_s$ is \emph{weakly opaque} with respect to non-secret initial state set $\mathcal{X}_{ns}^\prime\subseteq\mathcal{X}_{ns}$, if for \emph{some} $x_s(0)\in\mathcal{X}_s$, it holds that $x_s(0)\xrightarrow{\text{o}}\mathcal{X}_{ns}^\prime$. We denote this relation by $\mathcal{X}_{s}\xrightarrow{\text{w.o.}}\mathcal{X}_{ns}^\prime$.\hfill $\square$
\end{definition}
\begin{figure}[htp]
    \centering
    \includegraphics[width=7cm]{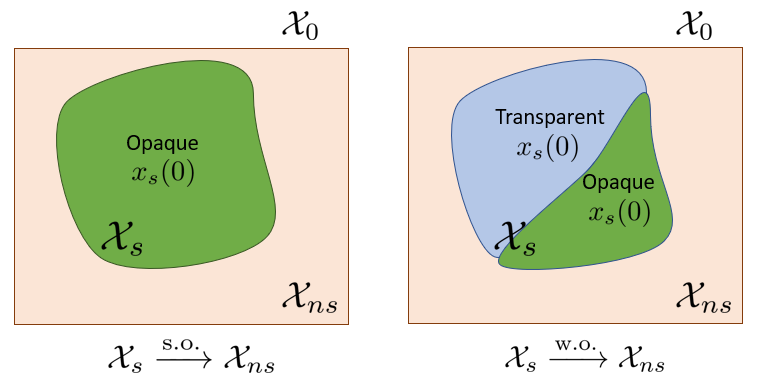}
    \caption{The left and right figures represent strongly and weakly opaque sets, respectively. Note that the green sets in both figures are (strongly) opaque.}
    \label{fig:Opacity_Definition}
\end{figure}
Strong (weak) opacity implies that all (some) states in the set $\mathcal{X}_s$ are opaque, as illustrated in Fig. \ref{fig:Opacity_Definition}.

\begin{remark}
The above definitions are also referred to as `initial state opacity' in some papers (e.g. definition III.1 in \cite{ApproxOpacity_IEEETAC_2021}). Further, these definitions differ from the definitions of $\mathcal{K}$-ISO used in \cite{OpacityLinearSystems_IEEETAC_2020}. In \cite{OpacityLinearSystems_IEEETAC_2020}, the outputs at \emph{specific time instants} corresponding to secret and non-secret initial states should match, whereas in our definitions, the whole \emph{output sequences} should match. \hfill $\square$
\end{remark}
Next, we define opacity ordering of two opaque sets. This will be used later to analyze the trade-off between opacity and attack detectability.
\begin{definition}[\emph{Opacity Ordering}]
Given two (strongly) opaque sets $\mathcal{X}_s^1$ and $\mathcal{X}_s^2$, we say that $\mathcal{X}_s^1$ is \emph{more opaque} than $\mathcal{X}_s^2$ if $\mathcal{X}_s^2\subset \mathcal{X}_s^1$. \hfill $\square$
\end{definition}

% {\color{red} Include the discussion on why we need to define these subsets and what do they represent.}
% For a system, given $\mathcal{X}_s$, let $\mathcal{X}_{ns}^o(\mathcal{X}_s)$ represent a subset of $\mathcal{X}_{ns}$ such that for all $x_{ns}^o(0)\in\mathcal{X}_{ns}^o(\mathcal{X}_s)$, $\mathcal{X}_s\xrightarrow{\text{w.o.}}x_{ns}^o(0)$. Similarly, given $\mathcal{X}_{ns}$, let $\mathcal{X}_{s}^o(\mathcal{X}_{ns})$ represent a subset of $\mathcal{X}_{s}$ such that $\mathcal{X}_{s}^o(\mathcal{X}_{ns})\xrightarrow{\text{s.o.}}\mathcal{X}_{ns}$. \par
% If parameters such as $\mathcal{X}_0$ or system matrices are modified, we say that opacity is increased if for an arbitrary set $\mathcal{X}_s$ or $\mathcal{X}_{ns}$, the corresponding sets $\mathcal{X}_{ns}^o(\mathcal{X}_s)$ and $\mathcal{X}_{s}^o(\mathcal{X}_{ns})$, respectively, under original parameters are a subset of $\mathcal{X}_{ns}^o(\mathcal{X}_s)$ and $\mathcal{X}_{s}^o(\mathcal{X}_{ns})$, respectively, under modified parameters.
%Finally, we can obtain the set of all output trajectories formed by these input sequences and initial states as: $\mathcal{Y}_{i}(k)=\bigcup\limits_{x_i(0)\in\mathcal{X}_i,U_i(k)\in\mathcal{U}_i(k)} Y_{x_i(0),U_i(k)}$. These terms are used to characterize the notion of opacity in the subsequent sections.\par

\noindent \emph{\bf Attack Model:} We consider an attacker\footnote{The attacker and the eavesdropper can be a single entity or two different entities.} that is capable of injecting malicious attack inputs in the actuators and modify sensor readings of the system $\Gamma$. Let the attack inputs be denoted by $\tilde{u}(k)$. We allow the attack inputs to be injected via channels that are different than the channels for normal inputs. We model this by using matrices $\tilde{B}$ and $\tilde{D}$ that can be different from $B$ and $D$.

%Now let us consider the system in the presence of an attacker who is capable of injecting attack inputs to the system to modify actuator and sensor readings. 

%The attack input $\tilde{u}(k)$ may be injected by the attacker through the same channels as that used by the operator or may also be injected through arbitrary channels. Therefore, the state matrix  $\tilde{B}$ and output matrix  $\tilde{D}$ through which attacks may be performed may or may not be related to $B$ and $D$ matrices. 

Since the normal input $u(k)$ is known to the operator, its effect may be eliminated for the purposes of attack detection. Therefore, we set $u(k)=0\enskip\forall \: k\geq 0$ for the attack model. The attack model (represented by $\tilde{\Gamma}$) is given as:
\begin{align*}
\tilde{\Gamma}\mathpunct{:}\quad\begin{aligned}\tilde{x}(k+1)&=A\tilde{x}(k)+\tilde{B}\tilde{u}(k),\\
\tilde{y}(k)&=C\tilde{x}(k)+\tilde{D} \tilde{u}(k),
\end{aligned}\numberthis\label{attack_model}
\end{align*}
where $\tilde{x}\in \mathbb{R}^n$ and $\tilde{y}\in\mathbb{R}^m $ denote the attacked states and outputs, respectively, and  $\tilde{u}\in \mathbb{R}^q$. Let $\tilde{U}(k)=\begin{bmatrix}\tilde{u}(0)^T & \tilde{u}(1)^T & \ldots & \tilde{u}(k)^T\end{bmatrix}^T$ denote the attack input sequence. Further, let $\tilde{Y}_{x(0),\tilde{U}(k)}$ denote the output sequence (vector) produced by applying the input sequence $\tilde{U}(k)$ to the initial state $x(0)$, which can  be expressed as:
\begin{align*}
\tilde{Y}_{x(0),\tilde{U}(k)}=O_kx(0)+F_k^{\tilde{\Gamma}}\tilde{U}(k),\numberthis \label{aop-seq}
\end{align*}
where $F_k^{\tilde{\Gamma}}$ is computed by replacing $B$ and $D$ by $\tilde{B}$ and $\tilde{D}$, respectively, in the expression for $F_k^{\Gamma}$ in \eqref{ip-op}.

We assume that the attacker knows the system matrices $A,B,C,D$ and the initial state set $\mathcal{X}_0$. Further, we assume that $[\tilde{B}^T,\tilde{D}^T]^T$ is full column rank. Next we present a definition for attacks that cannot be detected.

%The set of all output trajectories that can be produced by a system for this setting is denoted by $\tilde{\mathcal{Y}}(k)=\bigcup\limits_{x(0)\in\mathcal{X}_0,\tilde{U}(k)\in\tilde{\mathcal{U}}(k)} \tilde{Y}_{x(0),\tilde{U}(k)}$.\par

\begin{definition}[\emph{Undetectable Attacks} \cite{DynamicDetection_IEEETAC_2017}] \label{def:attack_det}
An attack $\tilde{U}(k)$ is said to be undetectable if there exist initial states $x(0),x^\prime(0)\in \mathcal{X}_0$ such that  $\tilde{Y}_{x(0),\tilde{U}(k)}=\tilde{Y}_{x^\prime (0),0}$, or equivalently, $\tilde{Y}_{x(0)-x^\prime (0),\tilde{U}(k)}=0$. We denote an undetectable attack sequence by $\tilde{U}_u(k)=\begin{bmatrix}\tilde{u}_u(0)^T & \tilde{u}_u(1)^T & \ldots & \tilde{u}_u(k)^T\end{bmatrix}^T$ and the set of all undetectable attacks by $\tilde{\mathcal{U}}_u(k)$.\hfill $\square$
\end{definition}

For undetectable attacks, the output produced by the system is same as the output produced by a zero attack input sequence (no attack) with appropriate initial conditions. Therefore, an attack detector\footnote{The attack detector is a dynamic detector as defined in \cite{DynamicDetection_IEEETAC_2017}, which operates on the entire output sequences.} that uses the outputs cannot determine whether the system is under attack or not. The existence of undetectable attacks depends on the weakly unobservable subspace of the system, which we define next.

\begin{definition}[\emph{Weakly Unobservable Subspace (WUS)} \cite{High-order_ECC_2007}] 
The weakly unobservable subspace of system \eqref{normal_model} (denoted by $\mathcal{V}(\Gamma)$) is defined as:
\begin{align*}
   \mathcal{V}(\Gamma)& =\{x\in\mathbb{R}^n:\exists\: U(k)\text{ such that }Y_{x,U(k)}=0,\:\forall\:k\geq 0\} \\
   &= \{x\in\mathbb{R}^n:\exists\: U(n-1)\text{ such that }Y_{x,U(n-1)}=0\},
\end{align*}
where the second equality follows from the Cayley-Hamilton Theorem.\hfill $\square$
\end{definition}

The subspace $\mathcal{V}(\tilde{\Gamma})$ is fundamentally connected to existence of undetectable attacks. In particular, it is known that if $\mathcal{V}(\tilde{\Gamma})\neq 0$, then there exists an attack sequence $\tilde{U}_u(k)$ in \eqref{attack_model} that is undetectable for all $k\geq 0$ \cite{AttackDetection_IEEETAC_2013}, \cite{High-order_ECC_2007}. In the next section, we show that $\mathcal{V}(\Gamma)$ is also connected to the opacity, and use this fact to characterize the trade-off between opacity and attack detectability in Section \ref{sec:tradeoffs}.

\section{Characterization of Opaque Sets} \label{sec:opacity}
We begin by defining conditions for opacity of \emph{individual} initial states, followed by the conditions for opacity of \emph{sets} of initial states.

\begin{lemma}\label{Opaque_Condition}
Given different initial states $x_s(0), x_{ns}(0)$, we have $x_s(0)\xrightarrow{\text{o}}\{x_{ns}(0)\}$ iff $x_s(0)-x_{ns}(0)\in \mathcal{V}(\Gamma)$.
\end{lemma}
\begin{proof}
If: From definition of $\mathcal{V}(\Gamma)$, there exists $U(k)$ such that $O_k(x_s(0)-x_{ns}(0))+F_k^\Gamma U(k)=0\:\forall k\geq 0$. Therefore, for all $U_s(k)$, $U_{ns}(k)$ can be chosen such that $U_s(k)-U_{ns}(k)=U(k)$. Consequently, we have for all $U_s(k)$, the following relation: $O_kx_s(0)+F_k^\Gamma U_s(k)=O_kx_{ns}(0)+F_k^\Gamma U_{ns}(k)$, which implies that $x_s(0)\xrightarrow{\text{o}}\{x_{ns}(0)\}$.\par
Only if: For any $k\geq 0$, we have that for all $U_s(k)$ there exists $U_{ns}(k)$ such that $O_kx_s(0)+F_k^\Gamma U_s(k)=O_kx_{ns}(0)+F_k^\Gamma U_{ns}(k)$ or $O_k(x_s(0)-x_{ns}(0))+F_k^\Gamma (U_s(k)-U_{ns}(k))=0$. From this, we can see that for all $k\geq 0$ there exists an input sequence $U_s(k)-U_{ns}(k)$ such that with the initial condition $x_s(0)-x_{ns}(0)$, the output is 0. Therefore, $x_s(0)-x_{ns}(0)\in\mathcal{V}(\Gamma)$
\end{proof}
\begin{corollary}\label{Corr_Opaque_State}
The following two statements hold true:\par
\noindent 1. Given $x_s(0)$, a state $x_{ns}(0)\neq x_{s}(0)$ satisfies $x_s(0)\xrightarrow{\text{o}}\{x_{ns}(0)\}$ iff $x_{ns}(0) \in x_{s}(0)\bigoplus\mathcal{V}(\Gamma)$.\par
\noindent 2. Given $x_{ns}(0)$, a state $x_s(0)\neq x_{ns}(0)$ satisfies $x_s(0)\xrightarrow{\text{o}}\{x_{ns}(0)\}$ iff $x_s(0)\in x_{ns}(0)\bigoplus\mathcal{V}(\Gamma)$.
\end{corollary}
Lemma \ref{Opaque_Condition} provides the necessary and sufficient condition for an initial state to be opaque and shows that opacity of initial states is fundamentally connected and completely determined by $\mathcal{V}(\Gamma)$. Further, Corollary \ref{Corr_Opaque_State} shows that the set of non-secret states that makes a secret state opaque (and vice-versa) is constrained by $\mathcal{V}(\Gamma)$. Next, we extend these results to specify conditions for weak and strong opacity of \emph{sets} of initial states.

\begin{lemma}\label{WISO_Condition}
Given non-empty and disjoint sets $\mathcal{X}_s$ and $\mathcal{X}_{ns}^\prime\subseteq\mathcal{X}_{ns}$, we have $\mathcal{X}_s\xrightarrow{\text{w.o.}}\mathcal{X}_{ns}^\prime$ iff $(\mathcal{X}_s\bigoplus -\mathcal{X}_{ns}^\prime)\cap \mathcal{V}(\Gamma)\neq \phi$.
\end{lemma}
\begin{proof}
If: Let  $(\mathcal{X}_s\bigoplus -\mathcal{X}_{ns}^\prime)\cap \mathcal{V}(\Gamma)\neq \phi$. Then there exists $x_s(0)\in \mathcal{X}_s$ and $x_{ns}(0)\in \mathcal{X}_{ns}^\prime$ such that $x_s(0)-x_{ns}(0)\in \mathcal{V}(\Gamma)$. Due to Lemma \ref{Opaque_Condition}, $x_s(0)\xrightarrow{\text{o}}\{x_{ns}(0)\}$, which implies that $\mathcal{X}_s\xrightarrow{\text{w.o.}}\mathcal{X}_{ns}^\prime$.\par
Only if: We prove by contrapositive argument. Let $(\mathcal{X}_s\bigoplus -\mathcal{X}_{ns}^\prime)\cap \mathcal{V}(\Gamma) = \phi$. Therefore there exists no pair $(x_s(0),x_{ns}(0))\in\mathcal{X}_s\times\mathcal{X}_{ns}^\prime$ such that $x_s(0)-x_{ns}(0)\in \mathcal{V}(\Gamma)$. Hence, by Lemma \ref{Opaque_Condition}, all $x_s(0)\in\mathcal{X}_s$ are transparent with respect to $\mathcal{X}_{ns}^\prime$. Therefore, it is not true that $\mathcal{X}_s\xrightarrow{\text{w.o.}}\mathcal{X}_{ns}^\prime$.
\end{proof}
\begin{lemma}\label{SISO_Condition}
Given non-empty and disjoint sets $\mathcal{X}_s$ and $\mathcal{X}_{ns}^\prime\subseteq\mathcal{X}_{ns}$, we have $\mathcal{X}_s\xrightarrow{\text{s.o.}}\mathcal{X}_{ns}^\prime$  iff $\mathcal{X}_s\subset \mathcal{X}_{ns}^\prime\bigoplus\mathcal{V}(\Gamma)$.
\end{lemma}
\begin{proof}
If: Since $\mathcal{X}_s\subset \mathcal{X}_{ns}^\prime\bigoplus\mathcal{V}(\Gamma)$, for all $x_s(0)\in \mathcal{X}_s$, there exists $x_{ns}(0)\in \mathcal{X}_{ns}^\prime$ such that $x_s(0)-x_{ns}(0)\in \mathcal{V}(\Gamma)$. Due to Lemma \ref{Opaque_Condition}, every $x_s(0)\in \mathcal{X}_s$ has a corresponding $x_{ns}(0)\in \mathcal{X}_{ns}^\prime$ such that $x_s(0)\xrightarrow{\text{o}}\{x_{ns}(0)\}$. This implies that $\mathcal{X}_s\xrightarrow{\text{s.o.}}\mathcal{X}_{ns}^\prime$.\par
Only if: We prove by contrapositive argument. Let $\mathcal{X}_s\supseteq \mathcal{X}_{ns}^\prime\bigoplus\mathcal{V}(\Gamma)$. We first consider $\mathcal{X}_s=\mathcal{X}_{ns}^\prime\bigoplus\mathcal{V}(\Gamma)$. This can be expanded as $\mathcal{X}_s=\mathcal{X}_{ns}^\prime\bigoplus 0\cup\mathcal{X}_{ns}^\prime\bigoplus (\mathcal{V}(\Gamma)\backslash 0)=\mathcal{X}_{ns}^\prime\cup\mathcal{X}_{ns}^\prime\bigoplus (\mathcal{V}(\Gamma)\backslash 0)$. However, since $\mathcal{X}_s$ and $\mathcal{X}_{ns}$ are disjoint, this is not possible and thus $\mathcal{X}_s\neq\mathcal{X}_{ns}^\prime\bigoplus\mathcal{V}(\Gamma)$. Let us then consider $\mathcal{X}_s\supset\mathcal{X}_{ns}^\prime\bigoplus\mathcal{V}(\Gamma)$. Therefore, $\mathcal{X}_s\cap (\mathcal{X}_{ns}^\prime\bigoplus\mathcal{V}(\Gamma))^c\neq\phi$. Thus, there exists at least one $x_s(0)\in \mathcal{X}_s$ such that there does not exist $x_{ns}(0)\in \mathcal{X}_{ns}^\prime$ for which  $x_s(0)-x_{ns}(0)\in \mathcal{V}(\Gamma)$. Hence, by Lemma \ref{Opaque_Condition}, such an $x_s(0)$ is transparent with respect to $\mathcal{X}_{ns}^\prime$. Therefore, it is not true that $\mathcal{X}_s\xrightarrow{\text{s.o.}}\mathcal{X}_{ns}^\prime$.
\end{proof}
Same as before, the conditions in Lemmas \ref{WISO_Condition} and \ref{SISO_Condition} are completely dependent on $\mathcal{V}(\Gamma)$.
\begin{remark}[\emph{Verifying Opacity Conditions}]
    The verification of opacity conditions in Lemmas \ref{WISO_Condition} and \ref{SISO_Condition} require  computation of Minkowski sum of sets. Algorithms to compute this for polytope sets in $\mathbb{R}^n$ are well developed in the literature, e.g. \cite{MinkowskiSum_JAMP_2014, MinkowskiSum_AICCMCGS_2015}. Also, Minkowski sum of sets with subspaces may be performed by computing the sum on the basis vectors of the subspace. An algorithm to find a basis for $\mathcal{V}(\Gamma)$ is found in \cite{Book_2002_ControlTheory}. The exact computation of these sets is not the focus of this paper and we defer this for future work. 
\end{remark}

Next, we analyze the largest possible opaque set for the system. Determining this largest set is important because it provides a fundamental limit beyond which a  larger opaque set cannot be constructed. To this aim, we formulate the following optimization problem:
\begin{equation}\label{eq_Xs_Opti}
\begin{aligned}
    \argmax_{\mathcal{X}_s} &\quad |\mathcal{X}_s| \\
    \text{ s.t. } & \quad  \mathcal{X}_s\xrightarrow{\text{s.o.}}\mathcal{X}_0\backslash\mathcal{X}_s.
\end{aligned}
\end{equation}

The next lemma provides a solution to the above optimization problem.

\begin{lemma}\label{max_sys-opacity}
For $\mathcal{X}_0=\mathbb{R}^n$, one solution to \eqref{eq_Xs_Opti} is $\mathcal{X}_s=\mathbb{R}^n\backslash\mathcal{V}(\Gamma)^{\perp}$, where $\mathcal{V}(\Gamma)^{\perp}$ is the orthogonal complement of $\mathcal{V}(\Gamma)$.
\end{lemma}
\begin{proof}
From Corollary \ref{Corr_Opaque_State}, for any $x_s(0)$, the corresponding $x_{ns}(0)$ such that $x_s(0)\xrightarrow{\text{o}}\{x_{ns}(0)\}$ belongs to the set $x_s(0)\bigoplus\mathcal{V}(\Gamma)$, which is the space formed by shifting $\mathcal{V}(\Gamma)$ parallel to itself to contain $x_s(0)$. Using these shifted spaces, we construct the largest opaque set.\par
If from each of such space we choose a single $x(0)$ to be $x_{ns}(0)$ and the remaining states to be $x_s(0)$, all the $x_s(0)$ in that space will be opaque due to the chosen $x_{ns}(0)$. This is the secret set with largest possible cardinality for this space. This secret set is also (strongly) opaque and its opacity is solely due to the chosen $x_{ns}(0)$ in this space. Finally, the union of all such spaces form $\mathcal{X}_0$. Therefore, choosing one $x_{ns}(0)$ from each space to form $\mathcal{X}_{ns}$ will result in $\mathcal{X}_s=\mathcal{X}_0\backslash\mathcal{X}_{ns}$ to be (strongly) opaque and to have the largest cardinality possible in the system. \par
One such $\mathcal{X}_{ns}$ is $\mathcal{V}(\Gamma)^{\perp}$, for which $\mathcal{X}_s=\mathbb{R}^n\backslash\mathcal{V}(\Gamma)^{\perp}$. This is because every vector of $\mathcal{V}(\Gamma)^{\perp}$ belongs to only one parallelly shifted space of $\mathcal{V}(\Gamma)$.
\end{proof}
Lemma \ref{max_sys-opacity} shows that the largest possible opaque set is constrained by $\mathcal{V}(\Gamma)$. This highlights the importance of $\mathcal{V}(\Gamma)$ in making the states opaque, and this is formalized in the following theorem.
\begin{theorem} \label{thm:opacity_WUS}
Consider two systems $\Gamma_1$ and $\Gamma_2$ and let $\mathcal{X}_0=\mathbb{R}^n$. For any (strongly) opaque set $\mathcal{X}_s^1$ in $\Gamma_1$, there exists a more (strongly) opaque set $\mathcal{X}_s^2$ in $\Gamma_2$ iff $\mathcal{V}(\Gamma_1)\subset\mathcal{V}(\Gamma_2)$.
\end{theorem}
\begin{proof}
Refer Appendix.
\end{proof}

Theorem \ref{thm:opacity_WUS} implies that expanding the subspace $\mathcal{V}(\Gamma)$ (by modifying the system matrices $A,B,C,D$) allows us to increase the size of any opaque set. This again highlights the fundamental connection between opacity and WUS. Next, we present an example to explain the results of this section.

\begin{example}\label{example1}
Consider the following system:
\begin{align*}
    x(k+1)&=\begin{bmatrix}
1 & 1\\
0 & 1
\end{bmatrix}
x(k)+\begin{bmatrix}
1 & 1\\
1 & 0
\end{bmatrix}u(k),\\
y(k)&=\begin{bmatrix}
1 & 0
\end{bmatrix}x(k).
\end{align*}
For this system, $\mathcal{V}(\Gamma)=\text{span}\{[0,1]^T\}$.

\noindent (i) We first consider opacity of individual initial states. Let $x_s(0)=[1,1]^T$ and $x_{ns}(0)=[1,0]^T$. Then, $x_s(0)-x_{ns}(0)=[0,1]^T\in\mathcal{V}(\Gamma)$. Thus, by Lemma \ref{Opaque_Condition}, $x_s(0)\xrightarrow{\text{o}} \{x_{ns}(0)\}$. We show this explicitly for $k=2$.
From opacity Definition \ref{def:opacity_state} , we have that for all $U_s(2)$, there should exist some $U_{ns}(2)$ such that  $Y_{x_s(0),U_s(2)}=Y_{x_{ns}(0),U_{ns}(2)}$. Using \eqref{nop-seq}, this is equivalent to $O_2x_s(0)+F_2^{\Gamma}U_s(2)=O_2x_{ns}(0)+F_2^{\Gamma}U_{ns}(2)$. Substituting $O_2,F_2^{\Gamma},x_{s}(0),x_{ns}(0)$ and rearranging, we get the following linear equation:\par 
{\small\begin{align*}
\underbrace{\begin{bmatrix}
0\\
1\\
2
\end{bmatrix}}_{O_2(x_s(0)-x_{ns}(0))}+\underbrace{\begin{bmatrix}
0 & 0 & 0 & 0 & 0 & 0\\
1 & 1 & 0 & 0 & 0 & 0\\
2 & 1 & 1 & 1 & 0 & 0
\end{bmatrix}}_{F_2^{\Gamma}}U_s(2) = F_2^{\Gamma}U_{ns}(2).
\end{align*}}
Since $O_2(x_s(0)-x_{ns}(0))\in\mathcal{R}(F_2^\Gamma)$, we see that for any $U_s(2)$ there exists $U_{ns}(2)$ that solves this equation.
For instance, both $U_s(2)=[1,1,1,1,1,1]^T$ and $U_{ns}(2)=[1,2,1,2,1,2]^T$ with corresponding initial states result in the output sequence $[1,4,8]^T$. 
%Since all $U_s(2)$ has a corresponding $U_{ns}(2)$ for which $\Gamma$ produces the same output, $x_s(0)$ cannot be estimated solely from the output sequences until $k=2$. Hence, opacity of $x_s(0)$ is seen until $k=2$.\par 
Further, since $x_s(0)\bigoplus\mathcal{V}(\Gamma)=x_{ns}(0)\bigoplus\mathcal{V}(\Gamma)=\{[1,c]^T:c\in\mathbb{R}\}$, it holds that $x_{ns}(0)\in x_s(0)\bigoplus\mathcal{V}(\Gamma)$ and $x_s(0)\in x_{ns}(0)\bigoplus\mathcal{V}(\Gamma)$. Therefore, $x_s(0)$ and $x_{ns}(0)$ satisfy Corollary \ref{Corr_Opaque_State}.\par
\noindent (ii) Next, we focus on opacity of sets. Let  $\mathcal{X}_0=\{x\in\mathbb{R}^2:||x||_{\infty}=1\}$ and let $\mathcal{X}_{ns}=\{[y,0]^T:y
\in [-1,1]\}$. We note that $(\mathcal{X}_{ns}\bigoplus\mathcal{V}(\Gamma))=\{[c,d]^T:c\in[-1,1],d\in\mathbb{R}\}$. Therefore, as seen in Fig. \ref{fig:Example1}, any $x_s(0)\in\mathcal{X}_s = \mathcal{X}_0\backslash\mathcal{X}_{ns}$ belongs to $\mathcal{X}_{ns}\bigoplus\mathcal{V}(\Gamma)$. Thus, $\mathcal{X}_s\subset \mathcal{X}_{ns}\bigoplus\mathcal{V}(\Gamma)$ and $\mathcal{X}_s\xrightarrow{\text{s.o.}}\mathcal{X}_{ns}$ as per Lemma \ref{SISO_Condition}.
\begin{figure}[htp]
    \centering
    \includegraphics[width=5cm]{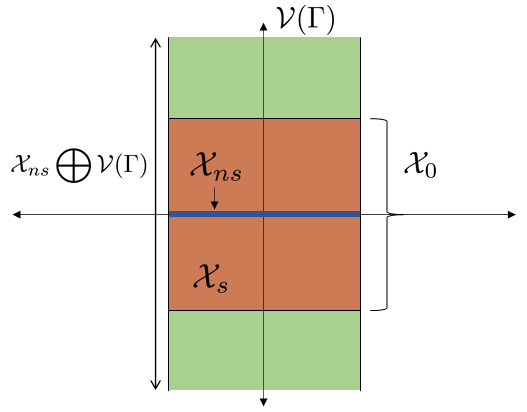}
    \caption{Pictorial representation of Example 1, part(ii). $\mathcal{X}_0$ is the brown square, $\mathcal{X}_{ns}$ is the blue line segment, $\mathcal{V}(\Gamma)$ is the $y$ axis, and $\mathcal{X}_{ns}\bigoplus\mathcal{V}(\Gamma)$ is the infinite green strip. Since $\mathcal{X}_s\subset\mathcal{X}_{ns}\bigoplus\mathcal{V}(\Gamma),$ we have $\mathcal{X}_s\xrightarrow{\text{s.o.}}\mathcal{X}_{ns}$.}
    \label{fig:Example1}
\end{figure}\par
\noindent (iii) Next, we find largest opaque set in $\Gamma$ with $\mathcal{X}_0=\mathbb{R}^2$. By Lemma \ref{max_sys-opacity}, the largest opaque set is $\mathcal{X}_s=\mathbb{R}^2\backslash\mathcal{V}(\Gamma)^\perp=\mathbb{R}^2\backslash$span$\{[1,0]^T\}$. Now, if we introduce matrix $D=[1,0]$ in the example, then $\mathcal{V}(\Gamma)$ expands to $\mathbb{R}^2$. Consequently, for this system, $\mathcal{X}_s=\mathbb{R}^2\backslash \{0\}$ is the largest opaque set, which is larger than the one for unmodified system.\hfill $\square$
\end{example}
%This section has laid a foundation for the notion of opacity and has built up different relations to the weakly unobservable subspace $\mathcal{V}(\Gamma)$. These primal results enable us to investigate a fundamental trade-off between opacity and attack detectability of a system. This we study in the next section.

\section{Opacity and Attack Detectability Trade-off} \label{sec:tradeoffs}
In this section, we use the relationship between opacity and WUS developed in Section \ref{sec:opacity} to characterize the trade-off between opacity and attack detectability. We do this in two ways by investigating the following questions:
\begin{itemize}
\item Does a system with opaque sets necessarily permit undetectable attacks? (Subsection \ref{sub:existence_tradeoff})
\item Does expanding opaque sets (by expanding $\mathcal{X}_0$) increases the set of undetectable attacks? (Subsection \ref{sub:X0_tradeoff})
\end{itemize}

\subsection{Coexistence of Opaque Sets and Undetectable Attacks}\label{sub:existence_tradeoff}
In this subsection, we show that existence of opaque sets implies existence of undetectable attacks, and vice-versa. We assume $\mathcal{X}_0=\mathbb{R}^n$  and begin the analysis by providing the condition for existence of opaque sets. 

\begin{lemma}\label{WO_iff_VGamma}
There exists a (strongly) opaque set  $\mathcal{X}_s$ for System $\Gamma$ in \eqref{normal_model} iff $\mathcal{V}(\Gamma)\neq 0$.
\end{lemma}
\begin{proof}
Only if: If there exist sets $\mathcal{X}_s,\mathcal{X}_{ns}$ such that $\mathcal{X}_s\xrightarrow{\text{s.o.}}\mathcal{X}_{ns}$, then there exists a pair $(x_s(0), x_{ns}(0))$ for which $x_s(0)\xrightarrow{\text{o}}\{x_{ns}(0)\}$. Due to Lemma \ref{Opaque_Condition}, $x_s(0)-x_{ns}(0)\in\mathcal{V}(\Gamma)$. Since $\mathcal{X}_s$ and $\mathcal{X}_{ns}$ are disjoint, $x_s(0)\neq x_{ns}(0)$ and thus, $\mathcal{V}(\Gamma)\neq 0$. \par
If: Let $\mathcal{V}(\Gamma)\neq 0$. Therefore, there exists $x(0)\neq 0$ and $U(k)$ such that $O_k x(0)+F_k^\Gamma U(k)=0$. Therefore we can construct disjoint sets $\mathcal{X}_s$ and $\mathcal{X}_{ns}$ such that for all $x_s(0)\in\mathcal{X}_s$, corresponding $x_{ns}(0)\in\mathcal{X}_{ns}$ is such that $x_s(0)-x_{ns}(0)=x(0)\in\mathcal{V}(\Gamma)$. Due to this, using Lemma \ref{Opaque_Condition}, $x_s(0)\xrightarrow{\text{o}}\{x_{ns}(0)\}$ for all $x_s(0)\in\mathcal{X}_s$. Consequently, $\mathcal{X}_s\xrightarrow{\text{s.o.}}\mathcal{X}_{ns}$.
\end{proof}
The following theorem shows that one cannot construct a system with opaque sets without permitting undetectable attacks.
\begin{theorem}\label{WO_iff_Attack}
The following two statements are equivalent:\par
\noindent 1. There exists a (strongly) opaque set $\mathcal{X}_s$ for $\Gamma$.\par
\noindent 2. There exists a $\tilde{\Gamma}$ (that is, a pair  $(\tilde{B},\tilde{D}$)), for which:\par
a. There exists an attack sequence $\tilde{U}_u(k)\neq 0$ that is undetectable for any $k\geq 0$,\par
b. $\mathcal{R}([\tilde{B}^{T},\tilde{D}^{T}]^T)\subseteq \mathcal{R}([B^T,D^T]^T)$. \par 
\end{theorem}
\begin{proof}
Refer Appendix.
\end{proof}
Theorem \ref{WO_iff_Attack} shows that one cannot have opacity in the system without making it inevitably vulnerable to undetectable attacks. Moreover, if all attacks are detectable for all ($\tilde{B},\tilde{D}$), then no opaque set can exist in the system. This implies that a fundamental trade-off exists between opacity and attack detectability for linear systems. Next, we characterize a set of attacked systems $\tilde{\Gamma}$ in Theorem \ref{WO_iff_Attack} for which undetectable attacks exist.

\begin{lemma}\label{Gamma_set}
The following two statements hold:\\
1. Suppose there exists a (strongly) opaque set for system $\Gamma$. Then, there exists $\tilde{U}_u(k)\neq 0$ undetectable for all $k\geq 0$ for a system $\tilde{\Gamma}$ if $\mathcal{R}([\tilde{B}^{T},\tilde{D}^{T}]^T) \supseteq \mathcal{R}([B^T,D^T]^T)$.

\noindent 2. Suppose there exists $\tilde{U}_u(k)\neq 0$ undetectable for all $k\geq 0$ for a system $\tilde{\Gamma}$. Then, there exists an opaque set for a system $\Gamma$ if $ \mathcal{R}([B^T,D^T]^T)  \supseteq \mathcal{R}([\tilde{B}^{T},\tilde{D}^{T}]^T) $.

% Given a system $\Gamma$ for which there exists a (strongly) opaque set,  there exist then for all $\tilde{\Gamma}$ for which $\mathcal{R}([B^T,D^T]^T)\subseteq \mathcal{R}([\tilde{B}^{T},\tilde{D}^{T}]^T)$ there exists $\tilde{U}_u(k)\neq 0$ undetectable for all $k\geq 0$ .\\
% 2. Given $\tilde{\Gamma}$ for which there exists $\tilde{U}_u(k)\neq 0$ undetectable for all $k\geq 0$, then for all $\Gamma$ for which $\mathcal{R}([\tilde{B}^{T},\tilde{D}^{T}]^T\subseteq \mathcal{R}([B^T,D^T]^T)$ there exists (strongly) opaque set.
\end{lemma}
\begin{proof}
    $\mathcal{V}(\Gamma)$ can be generated recursively as \cite{Book_2002_ControlTheory}: $\mathcal{V}_0=\mathcal{X}_0, \left(\begin{array}{l}
A \\
C
\end{array}\right)\mathcal{V}_{k+1}=\left[\left(\mathcal{V}_{k} \times 0\right)+\mathcal{R}\left(\begin{array}{l}
B \\
D
\end{array}\right)\right]$, where $\mathcal{V}_{k} \times 0$ denotes the cartesian product of subspaces $\mathcal{V}_k$ and 0.\par
Statement 1: If there exists (strongly) opaque set for $\Gamma$, then by Lemma \ref{WO_iff_VGamma}, $\mathcal{V}(\Gamma)\neq 0$. Now, for all $\tilde{\Gamma}$ for which $\mathcal{R}([B^T,D^T]^T)\subseteq \mathcal{R}([\tilde{B}^{T},\tilde{D}^{T}]^T)$, we see from the algorithm that $\mathcal{V}(\Gamma)\subseteq\mathcal{V}(\tilde{\Gamma})$.  Consequently, $\mathcal{V}(\tilde{\Gamma})\neq 0$ and thus, there exist $\tilde{U}_u(k)\neq 0\:\forall\: k\geq 0$ for all such $\tilde{\Gamma}$.\par
Statement 2: Proof is similar to proof for Statement 1.
\end{proof}
The above results elucidate that opaque sets and undetectable attacks co-exist in linear systems under certain conditions. This is further illustrated in the following example where the system is modified to eliminate undetectable attacks for a certain set of attacked systems $\tilde{\Gamma}$.\\~\par
\noindent\textbf{Example 1. (Continued)}
For the system in Example 1, since $\mathcal{V}(\Gamma)=\text{span}\{[0,1]^T\}$ contains elements other than the origin, there exist strongly opaque sets in $\Gamma$ (c.f. Lemma \ref{WO_iff_VGamma}). Examples of such sets were shown previously. Further, if we consider $\tilde{\Gamma}=\Gamma$, there exists an attack $\tilde{U}_u(k)\neq 0$ on the system starting from an initial condition $0\neq x(0)\in\mathcal{V}(\tilde{\Gamma})=\text{span}\{[0,1]^T\}$ that is undetectable for any $k\geq 0$. For instance, the attack  $\tilde{U}_u(2)=[0,-1,0,-1,0,-1]^T$ is undetectable until $k=2$.\par
Next, we modify the system in order to eliminate undetectable attacks. The modified system has the output dynamics:
\begin{align*}
    y(k)&=\begin{bmatrix}
1 & 0\\
0 & 1
\end{bmatrix}x(k).
\end{align*}
We argue that all attacks in the modified system are detectable. 
The modified system has $\mathcal{V}(\Gamma)=0$. If we consider a $\tilde{\Gamma}$ such that  $\mathcal{R}([\tilde{B}^{T},\tilde{D}^{T}]^T)\subseteq \mathcal{R}([B^T,D^T]^T)$, then  $\mathcal{V}(\tilde{\Gamma})=\mathcal{V}(\Gamma)=0$ (based on recursive algorithm stated in the proof of Lemma \ref{Gamma_set}). By Definition \ref{def:attack_det} undetectable attacks $\tilde{U}(k)$ should satisfy $\tilde{Y}_{x(0)-x^\prime (0),\tilde{U}(k)}=0$ for any two initial states $x(0)$ and $x^\prime(0)$. This is equivalent to $\tilde{Y}_{z(0),\tilde{U}(k)} = 0$ for some initial state $z(0)$, which yields $O_kz(0)+F^{\tilde{\Gamma}}_k\tilde{U}(k)=0.$ Since $\mathcal{V}(\tilde{\Gamma})=0$, the only initial condition for which the above equation would hold is $z(0)=0$. Thus, an undetectable attack should result in zero output sequence with zero initial condition. Since $[\tilde{B}^T,\tilde{D}^T]^T$ is full column rank and $\tilde{D}=0$, we have $\mathcal{N}(\tilde{B})=0$. Hence, if for some $k_0, \tilde{u}(k_0)\neq 0$, then $\tilde{x}(k_0+1)\neq 0$ and $\tilde{y}(k_0+1)\neq 0$, which implies that such an attack is detectable at $k_0+1$. Thus, all attacks are detectable. 
Moreover, since $\mathcal{V}(\Gamma)=0$, no opaque set exists (c.f. Lemma \ref{WO_iff_VGamma}). Therefore, we observe that eliminating undetectable attacks also eliminates opaque sets, indicating the trade-off between the two.

\subsection{Relation between Sizes of Opaque and  Undetectable Attacks Set}\label{sub:X0_tradeoff}
We examine the effect of expanding the opaque set on the size of undetectable attack set, and vice versa, and show that there exists a trade-off between the two. The variations in these sets is achieved by changing the initial state set $\mathcal{X}_0$.\footnote{Note that expanding $\mathcal{X}_s$ without changing $\mathcal{X}_0$ does not affect the undetectable attack set.} An expansion of $\mathcal{X}_0$ may be performed by the operator, for instance, to include larger set of opaque secret states.  
We assume $\mathcal{X}_0\subset\mathbb{R}^n$.\par
To assess this trade-off, we first characterize the set of undetectable attacks in terms of initial state set $\mathcal{X}_0$.
\begin{lemma}\label{set_undetectableattacks}
For all $k\geq 0$, the set of undetectable attacks $\tilde{\mathcal{U}}_u(k)$ is given by the set $\{\tilde{U}_u(k):F_k^{\tilde{\Gamma}}\tilde{U}_u(k)\in O_k (\mathcal{X}_0\bigoplus-\mathcal{X}_0)\}$.
\end{lemma}
\begin{proof}
From definition of undetectable attacks, we have for any undetectable attack $\tilde{U}_u(k)$ with $k\geq 0$, there exists $x(0),x^\prime (0)\in\mathcal{X}_0$ such that $O_kx(0)+F_k^{\tilde{\Gamma}}\tilde{U}_u(k)=O_kx^\prime (0)$. Therefore, $F_k^{\tilde{\Gamma}}\tilde{U}_u(k)=O_k(x^\prime (0)-x(0))$. Taking into consideration all combinations of $x^\prime(0)$ and $x(0)$, we have that $F_k^{\tilde{\Gamma}}\tilde{U}(k)\in O_k (\mathcal{X}_0\bigoplus-\mathcal{X}_0)$.
\end{proof}
The above lemma shows that $\tilde{\mathcal{U}}_u(k)$ depends on $\mathcal{X}_0$. Using this fact and the definition of opacity, the following theorem describes the trade-off when $\mathcal{X}_0$ is expanded.
\begin{theorem}\label{initialset_tradeoff1}
Consider initial state sets $\mathcal{X}_0^1\subset\mathcal{X}_0^2\subseteq\mathbb{R}^n$. Let $\tilde{\mathcal{U}}_{u}(k)^{\mathcal{X}_0}$ denote the set of undetectable attacks on a system with initial state set $\mathcal{X}_0$. Then the following statements hold true:\\
1. For any (strongly) opaque set $\mathcal{X}_s^1\subset\mathcal{X}_0^1$, there exists a  (strongly) opaque set $\mathcal{X}_s^2\subset\mathcal{X}_0^2$ such that $\mathcal{X}_s^1\subseteq\mathcal{X}_s^2$. Further, $\mathcal{X}_s^1\subset\mathcal{X}_s^2$ iff there exists $x(0)\in\mathcal{X}_0^2\backslash\mathcal{X}_0^1$ that satisfies $(x(0)\bigoplus\mathcal{V}(\Gamma))\cap\mathcal{X}_0^2\neq \{x(0)\}$.\\
2. The set of undetectable attacks are related as $\tilde{\mathcal{U}}_{u}(k)^{\mathcal{X}_0^1}\subseteq\tilde{\mathcal{U}}_{u}(k)^{\mathcal{X}_0^2}\:\forall k\geq 0$. Further, for any $k\geq 0$, $\tilde{\mathcal{U}}_{u}(k)^{\mathcal{X}_0^1}\subset\tilde{\mathcal{U}}_{u}(k)^{\mathcal{X}_0^2}$ iff there exists $z(0)\in(\mathcal{X}_0^2\bigoplus-\mathcal{X}_0^2)$ that satisfies $-O_kz(0)\in F_k^{\tilde{\Gamma}}(\mathbb{R}^{(k+1)q}\backslash\tilde{\mathcal{U}}_u(k)^{\mathcal{X}_0^1})$. For special case of $\tilde{D}$ being square and full rank, we always have $\tilde{\mathcal{U}}_{u}(k)^{\mathcal{X}_0^1}\subset\tilde{\mathcal{U}}_{u}(k)^{\mathcal{X}_0^2}\:\forall k\geq 0$.
\end{theorem}
\begin{proof}
Refer Appendix.    
\end{proof}
Theorem \ref{initialset_tradeoff1} shows that on expanding $\mathcal{X}_0$, the opaque and undetectable attack sets can expand or remain unchanged. This implies a non-strict trade-off between opacity and attack detectability. Further, the theorem provides conditions under which these sets expand, and a strict trade-off exists in this case.
Note that expanding $\mathcal{X}_0$ always expands the set of undetectable attacks for a certain set of $\tilde{\Gamma}$ but the set of opaque secret states expands only under specific conditions. Therefore, the additional initial states to expand $\mathcal{X}_0$ should be chosen carefully to satisfy the theorem conditions, so that expansion of $\mathcal{X}_0$ also results in the expansion of $\mathcal{X}_s$.

\section{Illustrative Example}
We consider an automotive vehicle application to illustrate our results. The identity of a certain set of vehicles is to be kept secret (for example, ambulances, money trucks etc.). Their secrecy is achieved by having other vehicles move along with them throughout the journey. In this case, since the location information is same for both sets of vehicles, an eavesdropper cannot extract the identity of the secret vehicles by viewing just the location information. We also consider an attacker that is capable of injecting malicious control inputs to the vehicles remotely. The scenario is illustrated in Fig. \ref{fig:Illustrative_Example}. A similar example to this is considered in \cite{OpacityLinearSystems_IEEETAC_2020} under a different framework for opacity.\par
Let the linear system model $\Gamma$ be:\par
{\small \begin{align*}
    \begin{bmatrix}
p_x(k+1) \\
v_x(k+1)
\end{bmatrix}&=\! \begin{bmatrix}
1 & 1\\
0 & 1\\
\end{bmatrix}\begin{bmatrix}
p_x(k)\\
v_x(k)\\
\end{bmatrix}\!+\!\begin{bmatrix}
0.5\\
1\\
\end{bmatrix}
a_x(k),\\
y(k)&=\begin{bmatrix}
1 & 0
\end{bmatrix}\begin{bmatrix}
p_x(k)\\
v_x(k)\\
\end{bmatrix},
\end{align*}}

 \noindent where $p_x,v_x$ and $a_x$ represent the horizontal position, horizontal velocity and horizontal acceleration, respectively, of the vehicles.
 The non-secret vehicles are always mobile (along different circuits) while the secret vehicles start from rest to join the non-secret vehicles. Hence, secret and non-secret vehicles are differentiated based on their initial velocity.\par 
 Let $\mathcal{X}_0=\mathbb{R}^2$. Here, $\Gamma$ is observable and $\mathcal{V}(\Gamma)=$ span$\{[0,1]^T\}$. Therefore, for any non-secret vehicle with non-zero velocity at a particular position allowed by $\mathcal{X}_0$, any secret vehicle starting at rest from the same position will be opaque with respect to the non-secret vehicle. For instance,  $x_s(0)=[1,0]^T$, with an acceleration of $U_s(3)=[2,2,2,2]^T$ and $x_{ns}(0)=[1,1]^T$, with an acceleration of $U_{ns}(3)=[0,4,0,0]^T$ produces the same trajectory of $[p_x(0),p_x(1),p_x(2),p_x(3)]^T=[1,2,5,10]^T$.\par
 Since opaque sets exist in the system, undetectable attacks also co-exist. Consider $\tilde{\Gamma}=\Gamma$. An undetectable attack on the non-secret vehicle with state $x_s(0)=[1,0]^T$ is $\tilde{U}_u(3)=[2,-2,2,0]^T$, which is injected in addition to the normal input acceleration of $U_s(3)=[2,2,2,2]^T$. Due to this attack, the secret vehicle speeds up to the trajectory $[1,3,7,13]^T$ without getting detected.\par
 \begin{figure}[htp]
    \centering
    \includegraphics[width=8.5cm]{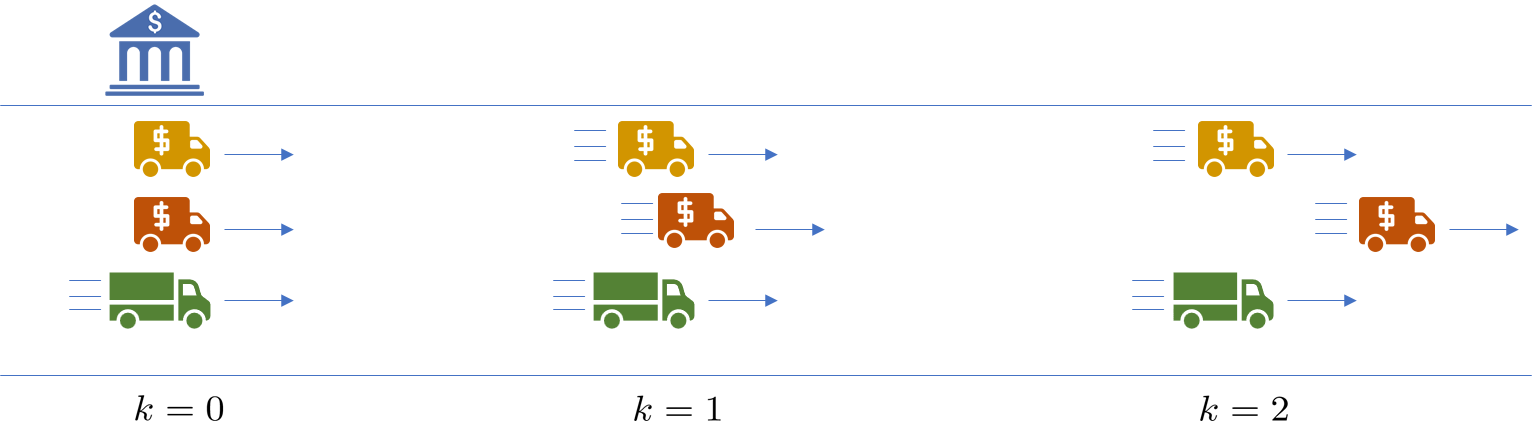}
    \caption{Trucks carrying money (denoted by yellow and red color) are secret vehicles. The non-secret vehicle (green color) moves along to keep both money trucks confidential. However, a remote undetectable attack on the red money truck causes it to speed up.}
    \label{fig:Illustrative_Example}
\end{figure}
\section{Conclusion}
We analyzed the underlying connection between the notion of opacity and attack detectability for linear systems. The fundamental relation between opacity and the weakly unobservable subspace was studied from multiple perspectives. Using this characterization, we showed that a fundamental trade-off exists between opaque sets and undetectable attacks. In particular, we showed that opaque sets co-exist with undetectable attack sets, and under certain conditions, expanding one results in the expansion of the other. %This result equips operators to modify system parameters such that the effect of the trade-off is minimized.\par
%This trade-off may be applied to a broad range of systems, especially in networked control systems that are gaining prevalence in recent times. One such example of a remotely controlled automotive system was considered.\par
Future directions include investigating the effect of changing systems matrices $A$ and $C$ on opacity and attack detectability, exploring numerical algorithms to efficiently verify the opacity conditions, and determining if attack inputs have any effect of opacity of certain states.

%In this paper, analysis of the trade-off by changing system matrices was not considered. Also, the trade-off was assessed independently for opaque sets and undetectable attacks when system parameters are modified. However, if these are considered together, attacks could also have an effect on opaque sets directly. As part of future work, these two problems could also be taken up to find further connections.
\section{Appendix}
\subsection{Proof of Theorem 1}
If: Since $\mathcal{V}(\Gamma_1)\subset\mathcal{V}(\Gamma_2)$, $\mathcal{V}(\Gamma_2)$ is a space of dimension higher than $\mathcal{V}(\Gamma_1)$. Consider any opaque set $\mathcal{X}_s^1$ in $\Gamma_1$. Since $\mathcal{X}_s^1$ is (strongly) opaque, it also holds that $\mathcal{X}_s^1\xrightarrow{\text{s.o.}}\mathcal{X}_{0}\backslash\mathcal{X}_{s}^1$. From Corollary \ref{Corr_Opaque_State}, for any $x_s(0)\in\mathcal{X}_s^1$, corresponding $x_{ns}(0)$ that makes it opaque belongs to the set $x_s(0)\bigoplus\mathcal{V}(\Gamma_1)$. Since $\mathcal{V}(\Gamma_1)\subset\mathcal{V}(\Gamma_2)$, it also holds that $x_{ns}(0)\in x_s(0)\bigoplus\mathcal{V}(\Gamma_2)$. Therefore, all $x_s(0)$ in $\Gamma_1$ will remain opaque in $\Gamma_2$ due to corresponding $x_{ns}(0)\in x_s(0)\bigoplus\mathcal{V}(\Gamma_1)$. \par
Now, for $\Gamma_2$, in order to construct $\mathcal{X}_s^2$, we need to have an extra opaque $x_s(0)\notin\mathcal{X}_s^1$. For this, we choose any $x_{ns}(0)\in\mathcal{X}_{0}\backslash\mathcal{X}_{s}^1$ and set it as a secret state $x_s(0)$. For this $x_s(0)$, the corresponding non-secret state, denoted by $x_{ns}^o(0)$, that will make it opaque in $\Gamma_2$ must be from the set $x_s(0)\bigoplus\mathcal{V}(\Gamma_2)$. Since $\mathcal{V}(\Gamma_1)\subset\mathcal{V}(\Gamma_2)$, we have that $x_s(0)\bigoplus\mathcal{V}(\Gamma_2)=\bigcup\limits_{x(0)}x(0)\bigoplus\mathcal{V}(\Gamma_1)$ for some initial states $x(0)\in\mathcal{X}_0$. Now each space $x(0)\bigoplus\mathcal{V}(\Gamma_1)$, except for $x_s(0)\bigoplus\mathcal{V}(\Gamma_1)$, will have at least one $x_{ns}(0)$. This is because if such an $x_{ns}(0)$ does not exist, all elements in the corresponding $x(0)\bigoplus\mathcal{V}(\Gamma_1)$ will be secret states that are not opaque in $\Gamma_1$. However, in $\mathcal{X}_s^1$, all secret states are opaque. Consequently, there exists some $x_{ns}(0)\in x_s(0)\bigoplus\mathcal{V}(\Gamma_2)$. Choosing this $x_{ns}(0)$ as $x_{ns}^o(0)$ will make $x_s(0)$ opaque in $\Gamma_2$. Also, since $\mathcal{V}(\Gamma_1)\subset\mathcal{V}(\Gamma_2)$, all other secret states in $x_s(0)\bigoplus\mathcal{V}(\Gamma_1)$ remain opaque in $\Gamma_2$ due to $x_{ns}^o(0)$.\par
Finally, since this $x_s(0)$ is an addition to $\mathcal{X}_s^1$ and is opaque, $\mathcal{X}_s^2$ is more (strongly) opaque than $\mathcal{X}_s^1$.\par
Only if: We prove by contrapositive argument. Therefore, for $\mathcal{V}(\Gamma_1)\supseteq\mathcal{V}(\Gamma_2)$, it is to be shown that there exists an opaque set $\mathcal{X}_s^1$ in $\Gamma_1$ such that there does not exist an opaque set $\mathcal{X}_s^2$ in $\Gamma_2$ which is more opaque than $\mathcal{X}_s^1$.\par
Let us consider such an opaque set $\mathcal{X}_s^1$ to be $\mathbb{R}^n\backslash\mathcal{V}(\Gamma_1)^\perp$. Therefore, $\mathcal{X}_{ns}^1=\mathcal{V}(\Gamma_1)^{\perp}$. First, let $\mathcal{V}(\Gamma_2)=\mathcal{V}(\Gamma_1)$. Now, in order to construct $\mathcal{X}_s^2$ that is more opaque than $\mathcal{X}_s^1$, we need to add opaque secret state not in $\mathcal{X}_s^1$. For this, at least one $x_{ns}(0)\in\mathcal{X}_{ns}^1$ has to be changed to $x_s(0)$. However, since $\mathcal{X}_{ns}^1=\mathcal{V}(\Gamma_1)^{\perp}$ and due to Corollary \ref{Corr_Opaque_State}, we have that for each $x_{ns}(0)\in\mathcal{X}_{ns}^1$, all elements  in $x_{ns}(0)\bigoplus\mathcal{V}(\Gamma_1)$ except $x_{ns}(0)$ are opaque secret states and are opaque solely due to $x_{ns}(0)$. Therefore, changing any $x_{ns}(0)$ to $x_s(0)$ would make all $x_s(0)$ in $x_{ns}(0)\bigoplus\mathcal{V}(\Gamma_1)$ transparent. If we consider $\Gamma_2$ such that $\mathcal{V}(\Gamma_2)\subset\mathcal{V}(\Gamma_1)$, we observe the same result. Therefore, for this set, additional opaque $x_s(0)$ cannot be added and thus such an $\mathcal{X}_s^2$ cannot exist.
\subsection{Proof of Theorem 2}
Only if: Since there exists (strongly) opaque set $\mathcal{X}_s$ in $\Gamma$, due to Lemma \ref{WO_iff_VGamma}, we have that $\mathcal{V}(\Gamma)\neq 0$ . If we take $\tilde{\Gamma}=\Gamma$, we get $\mathcal{V}(\tilde{\Gamma})\neq 0$ and therefore, there exists an attack sequence $\tilde{U}_u(k)$ that is undetectable for any $k\geq 0$. Now we show this attack is non-zero. Since $\Gamma$ is observable, $\tilde{\Gamma}$ is also observable. Therefore, from \eqref{aop-seq}, we have for $0\neq x(0)\in\mathcal{V}(\tilde{\Gamma})$, $F_k^{\tilde{\Gamma}} \tilde{U}_u(k)=-O_k x(0)\neq 0$. Hence, $\tilde{U}_u(k)\neq 0$ for any $k\geq 0$. Finally, since $\tilde{\Gamma}=\Gamma$, we have $\mathcal{R}([\tilde{B}^{T},\tilde{D}^{T}]^T)= \mathcal{R}([B^T,D^T]^T)$.\par
If: The proof is inspired by proof of Theorem 3 of \cite{DynamicDetection_IEEETAC_2017}. For any $k\geq 0$, given such an undetectable attack $\tilde{U}_u(k)$, there exists a corresponding $x(0)$ such that $O_k x(0)+F_k^{\tilde{\Gamma}}\tilde{U}_u(k)=0$ for any $k\geq 0$. Let $x(0)\neq 0$. Since $\tilde{U}_u(k)$ produces zero output sequence for all $k\geq 0$, $0\neq x(0)\in\mathcal{V}(\tilde{\Gamma})$. Therefore,  $\mathcal{V}(\tilde{\Gamma})\neq 0$.\par 
For the case when $x(0)=0$, since $\tilde{U}_u(k)\neq 0$, there exists minimum $k_{\text{min}}\geq 0$ such that  $\tilde{u}_u(k)=0\:\forall k<k_{\text{min}}$ and $\tilde{u}_u(k_{\text{min}})\neq 0$. Therefore,  $x(k)=0\:\forall 0\leq k\leq k_{\text{min}}$. Also, $x(k_{\text{min}}+1)=\tilde{B}\tilde{u}_u(k_{\text{min}})$. Since $\tilde{U}_u(k)$ is an undetectable attack, $y(k)=0\:\forall k\geq 0$. Consequently, $ y(k_{\text{min}})=0$ and thus, $\tilde{D}\tilde{u}_u(k_{\text{min}})=0$. Since $[\tilde{B}^T,\tilde{D}^T]^T$ is full column rank and since $\tilde{u}_u(k_{\text{min}})\neq 0$,  $\tilde{B}\tilde{u}_u(k_{\text{min}})\neq 0$. Therefore, $x(k_{\text{min}}+1)\neq 0$. Since $\tilde{U}_u(k)$ produces zero output for all $k\geq 0$, $\tilde{U}_u(k)$ produces zero output for all $k\geq k_{\text{min}}+1$. This implies that $x(k_{\text{min}}+1)\in\mathcal{V}(\tilde{\Gamma})$. Since $x(k_{\text{min}}+1)\neq 0$, we have $\mathcal{V}(\tilde{\Gamma})\neq 0$.\par
For both cases of $x(0)$ analyzed above, we have $\mathcal{V}(\tilde{\Gamma})\neq 0$. Now we show that if, in addition, $\tilde{\Gamma}$ is related to $\Gamma$ such that $\mathcal{R}([\tilde{B}^{T},\tilde{D}^{T}]^T)\subseteq \mathcal{R}([B^T,D^T]^T)$, then $\mathcal{V}(\Gamma)\neq 0$. Using recursive algorithm stated in the proof of Lemma \ref{Gamma_set}, we see that if $\tilde{\Gamma}$ is such that $\mathcal{R}([\tilde{B}^{T},\tilde{D}^{T}]^T)\subseteq \mathcal{R}([B^T,D^T]^T)$, then $\mathcal{V}(\tilde{\Gamma})\subseteq\mathcal{V}(\Gamma)$. Since $\mathcal{V}(\tilde{\Gamma})\neq 0$, we have $\mathcal{V}(\Gamma)\neq 0$. From this and Lemma \ref{WO_iff_VGamma}, there exist a (strongly) opaque set $\mathcal{X}_s$ for $\Gamma$.
\subsection{Proof of Theorem 3}
Statement 1: In order to construct $\mathcal{X}_s^2$, for which all its elements are opaque, $x(0)\in\mathcal{X}_0^2\backslash\mathcal{X}_0^1$ may be chosen as $x_{ns}(0)$ or as $x_s(0)$. Since all elements in $\mathcal{X}_s^1$ are opaque, $x(0)$ chosen as $x_{ns}(0)$ will not add any opaque $x_s(0)$ to $\mathcal{X}_s^1$.  Therefore, $\mathcal{X}_s^2=\mathcal{X}_s^1$ and $\mathcal{X}_s^2\xrightarrow{\text{s.o.}}\mathcal{X}_{0}^2\backslash\mathcal{X}_s^2$. However, if $x(0)$ is chosen as opaque $x_s(0)$, $\mathcal{X}_s^2=\mathcal{X}_s^1\cup x_s(0)$ is such that $\mathcal{X}_s^2\xrightarrow{\text{s.o.}}\mathcal{X}_{0}^2\backslash\mathcal{X}_s^2,$ and $\mathcal{X}_s^1\subset\mathcal{X}_s^2$. Hence, $\mathcal{X}_s^1\subseteq\mathcal{X}_s^2$. Next, we proceed to prove condition for $\mathcal{X}_s^1\subset\mathcal{X}_s^2$.\par
    Only if: We prove by contrapositive argument. Let $(x(0)\bigoplus\mathcal{V}(\Gamma))\cap\mathcal{X}_0^2=\{x(0)\}$. Since we consider $\mathcal{X}_s^1$ to be strongly opaque, if $x(0)$ is taken as $x_{ns}(0)$, it will not add additional opaque secret states. Therefore, let us consider $x(0)$ to be secret state $x_s(0)$. From Corollary \ref{Corr_Opaque_State}, we see that the corresponding $x_{ns}(0)$ such that $x_s(0)\xrightarrow{\text{o}}\{x_{ns}(0)\}$ is equal to $(x(0)\bigoplus\mathcal{V}(\Gamma))\cap\mathcal{X}_0^2$ which is equal to $x(0)$. However, this is not possible since $x_s(0)$ and $x_{ns}(0)$ must be different. Thus, if $x(0)$ is taken as secret state, it will not be opaque. Therefore, $\mathcal{X}_0^2$ does not add any opaque secret state $x_s(0)$ beyond what was present in $\mathcal{X}_0^1$.\par
If: Since there exists $x(0)\in\mathcal{X}_0^2\backslash\mathcal{X}_0^1$ such that $(x(0)\bigoplus\mathcal{V}(\Gamma))\cap\mathcal{X}_0^2\neq \{x(0)\}$, there exists $x^\prime (0)\neq x(0)$ for which $x^\prime(0)\in (x(0)\bigoplus\mathcal{V}(\Gamma))\cap\mathcal{X}_0^2$. This $x^\prime(0)$ can belong to either $\mathcal{X}_0^2\backslash\mathcal{X}_0^1$ or to $\mathcal{X}_0^1$. We consider the following cases:\par
1. $x^\prime(0)\in\mathcal{X}_0^2\backslash\mathcal{X}_0^1$. If $x(0)$ is chosen as $x_s(0)$ and $x^\prime(0)$ is chosen as $x_{ns}(0)$, then since $x^\prime(0)\in (x(0)\bigoplus\mathcal{V}(\Gamma))\cap\mathcal{X}_0^2$, due to Corollary \ref{Corr_Opaque_State}, $x(0)$ is an opaque secret state.\par
2. $x^\prime(0)\in\mathcal{X}_0^1$. In this case, since $\mathcal{X}_s^1\xrightarrow{\text{s.o.}}\mathcal{X}_{0}^1\backslash\mathcal{X}_{s}^1$, the $x^\prime(0)$ is $x_{ns}(0)$ or opaque $x_s(0)$ in the original set $\mathcal{X}_0^1$. These subcases are analyzed below:\par
a. $x^\prime(0)$ is $x_{ns}(0)\in\mathcal{X}_0^1$. In this case, if $x(0)$ is chosen as $x_s(0)$, then as before, since $x^\prime(0)\in (x(0)\bigoplus\mathcal{V}(\Gamma))\cap\mathcal{X}_0^2$, due to Corollary \ref{Corr_Opaque_State}, $x(0)$ is an opaque secret state.\par
b. $x^\prime(0)$ is opaque $x_s(0)\in\mathcal{X}_s^1$. Let $x_s^1(0)$ denote this $x_s(0)$. In this case, there exists a corresponding $x_{ns}^1(0)\in\mathcal{X}_0^1$ such that $x_s^1(0)\xrightarrow{\text{o}}\{x_{ns}^1(0)\}$. Using Corollary \ref{Corr_Opaque_State}, $x_{ns}^1(0)\in (x_s^1(0)\bigoplus\mathcal{V}(\Gamma))\cap\mathcal{X}_0^1$. If $x(0)$ is chosen as secret state, denoted by $x_s^2(0)$, then the corresponding $x_{ns}^2(0)$ such that $x_s^2(0)\xrightarrow{\text{o}}\{x_{ns}^2(0)\}$ should belong to the set $(x_s^2(0)\bigoplus\mathcal{V}(\Gamma))\cap\mathcal{X}_0^2$.\par
Our aim is to find $x_{ns}^2(0)$. We note that $x_{ns}^1(0)\in (x_s^1(0)\bigoplus\mathcal{V}(\Gamma))\cap\mathcal{X}_0^1\subseteq (x_s^1(0)\bigoplus\mathcal{V}(\Gamma))\cap\mathcal{X}_0^2$. If it can be shown that  $x_s^1(0)\bigoplus\mathcal{V}(\Gamma)=x_s^2(0)\bigoplus\mathcal{V}(\Gamma)$, we may choose $x_{ns}^1(0)$ as the required $x_{ns}^2(0)$ that will make $x_s^2(0)$ opaque. This can be shown as follows. Since from theorem, $x^\prime(0)\in (x(0)\bigoplus\mathcal{V}(\Gamma))\cap\mathcal{X}_0^2$, we have $x^\prime(0)\in (x(0)\bigoplus\mathcal{V}(\Gamma))$. Using notation introduced above, this is same as $x_s^1(0)\in (x_s^2(0)\bigoplus\mathcal{V}(\Gamma))$. Now we find $(x_s^2(0)\bigoplus\mathcal{V}(\Gamma))$. Let $V$ be the matrix with column vectors being the basis vectors of $\mathcal{V}(\Gamma)$. Then for some $y_0$, $x_s^1(0)=x_s^2(0)+Vy_0$. Therefore, for $\mathcal{V}(\Gamma)=Vx$ where $x\in\mathbb{R}^n$,  $x_s^1(0)\bigoplus\mathcal{V}(\Gamma)=x_s^2(0)+Vy_0+Vx=x_s^2(0)+V(x+y_0)=x_s^2(0)\bigoplus\mathcal{V}(\Gamma)$ (since $Vy_0\in\mathcal{V}(\Gamma)$). Thus, $x_s^1(0)\bigoplus\mathcal{V}(\Gamma)=x_s^2(0)\bigoplus\mathcal{V}(\Gamma)$.\par
Since $x_{ns}^1(0)\in (x_s^1(0)\bigoplus\mathcal{V}(\Gamma))\cap\mathcal{X}_0^2$, and since $x_s^1(0)\bigoplus\mathcal{V}(\Gamma)=x_s^2(0)\bigoplus\mathcal{V}(\Gamma)$, it holds that $x_{ns}^1(0)\in (x_s^2(0)\bigoplus\mathcal{V}(\Gamma))\cap\mathcal{X}_0^2$, which is the set to which $x_{ns}^2(0)$ should belong to so that $x_s^2(0)=x(0)$ is opaque. Since $x_{ns}^1(0)$ exists in $\mathcal{X}_0^1$, we see that $x(0)$ is an opaque secret state.\par
In all the former cases, it is seen that there exists $x(0)\in\mathcal{X}_0^2\backslash\mathcal{X}_0^1$ that may be chosen as an opaque secret state. This implies that there is an addition of an opaque $x_s(0)$ due to expansion of $\mathcal{X}_0^1$ to $\mathcal{X}_0^2$. Therefore $\mathcal{X}_s^2$ is more opaque than $\mathcal{X}_s^1$.\par
    Statement 2: By changing $\mathcal{X}_0$ such that $\mathcal{X}_0^1\subset\mathcal{X}_0^2$ in the relation in Lemma \ref{set_undetectableattacks}, we note that $\tilde{\mathcal{U}}_{u}(k)^{\mathcal{X}_0^1}\subseteq\tilde{\mathcal{U}}_{u}(k)^{\mathcal{X}_0^2}\:\forall k\geq 0$. Next we show condition for $\tilde{\mathcal{U}}_{u}(k)^{\mathcal{X}_0^1}\subset\tilde{\mathcal{U}}_{u}(k)^{\mathcal{X}_0^2}$. $\tilde{\mathcal{U}}_{u}(k)^{\mathcal{X}_0^1}\subset\tilde{\mathcal{U}}_{u}(k)^{\mathcal{X}_0^2}$ if and only if there exists $\tilde{U}_u(k)^{\mathcal{X}_0^2}\in\tilde{\mathcal{U}}_{u}(k)^{\mathcal{X}_0^2}$ such that $\tilde{U}_u(k)^{\mathcal{X}_0^2}\in\mathbb{R}^{(k+1)q}\backslash\tilde{\mathcal{U}}_u(k)^{\mathcal{X}_0^1}$. Also, from basic definition, for any $k\geq 0$, $\tilde{U}_u(k)^{\mathcal{X}_0^2}$ is undetectable if and only if $-O_k(x(0)-x^\prime (0))=F_k^{\tilde{\Gamma}}\tilde{U}_u(k)^{\mathcal{X}_0^2}$ for some $x(0),x^\prime(0)\in\mathcal{X}_0^2$. Let $z(0)=x(0)-x^\prime (0)$. We use these basic results in the following proof.\par
    Only if: From the above, it follows that $z(0)\in\mathcal{X}_0^2\bigoplus-\mathcal{X}_0^2$ and for any $k\geq 0$, $-O_kz(0)\in F_k^{\tilde{\Gamma}}(\mathbb{R}^{(k+1)q}\backslash\tilde{\mathcal{U}}_u(k)^{\mathcal{X}_0^1})$\par
    If: For some $x(0),x^\prime(0)\in\mathcal{X}_0^2$, there exists $z(0)=x(0)-x^\prime(0)$ such that for any $k\geq 0$, we have for some $\tilde{U}(k)^{\mathcal{X}_0^2}\in\mathbb{R}^{(k+1)q}\backslash\tilde{\mathcal{U}}_u(k)^{\mathcal{X}_0^1}$ it holds that $-O_kz(0)= F_k^{\tilde{\Gamma}}\tilde{U}(k)^{\mathcal{X}_0^2}$. From the above, $\tilde{U}(k)^{\mathcal{X}_0^2}$ is undetectable and $\tilde{\mathcal{U}}_{u}(k)^{\mathcal{X}_0^1}\subset\tilde{\mathcal{U}}_{u}(k)^{\mathcal{X}_0^2}$.
    
    Finally, if $\tilde{\Gamma}$ is chosen such that $\tilde{D}$ is square full rank, then $F_k^{\tilde{\Gamma}}$ is also square and full rank for all $k\geq 0$ (minimum rank of lower triangular matrix is sum of rank of diagonal blocks). Therefore, $\tilde{\mathcal{U}}_u(k)={(F_k^{\tilde{\Gamma}})}^{-1} O_k (\mathcal{X}_0\bigoplus-\mathcal{X}_0)\:\forall k\geq 0$. Consequently, if $\mathcal{X}_0^1\subset\mathcal{X}_0^2$, since $\tilde{\Gamma}$ is observable, it holds that $\tilde{\mathcal{U}}_{u}(k)^{\mathcal{X}_0^1}\subset\tilde{\mathcal{U}}_{u}(k)^{\mathcal{X}_0^2}\:\forall k\geq 0$.
\addtolength{\textheight}{-12cm}   % This command serves to balance the column lengths

\end{document}